\documentclass[journal]{IEEEtran}  

\usepackage[T1]{fontenc}
\usepackage[utf8]{inputenc}  

\usepackage{amsmath, amsthm, amssymb, amsfonts, bm, mathtools, dsfont, nccmath}
\theoremstyle{plain}

\theoremstyle{plain}

\usepackage{graphicx, epsfig, epstopdf}
\usepackage{caption, subcaption}
\usepackage{float, stfloats}  

\usepackage{algorithm}
\usepackage{algpseudocode}

\usepackage{tabularx}
\usepackage{enumerate}
\usepackage{xspace}
\usepackage{enumitem}

\usepackage[sort,compress]{cite}      
\usepackage[hidelinks]{hyperref}      

\theoremstyle{plain}
\newtheorem{theorem}{Theorem}

\newtheorem{lemma}{Lemma}

\usepackage{balance}

\usepackage{xcolor}

\usepackage[framemethod=tikz]{mdframed}
\definecolor{mycolor}{rgb}{0.122, 0.435, 0.698}
\newmdenv[innerlinewidth=0.5pt, roundcorner=4pt, linecolor=mycolor,
innerleftmargin=6pt, innerrightmargin=6pt,
innertopmargin=6pt, innerbottommargin=6pt]{mybox}

\newcommand{\bmh}{\mathbf{h}}  
\newcommand{\bmk}{\mathbf{k}}  
  
  \newcommand{\bmV}{\mathbf{V}}
\newcommand{\bmv}{\mathbf{v}}  
  
  \newcommand{\bmW}{\mathbf{W}}

  \newcommand{\bmr}{\mathbf{r}}

\newcommand{\bmw}{\mathbf{w}}
\usepackage[normalem]{ulem}
\usepackage[user]{zref}

\newcounter{revc}
\makeatletter
\zref@newprop{revcontent}{} \zref@addprop{main}{revcontent}
\zref@newprop{revsec}{}     \zref@addprop{main}{revsec}
\zref@newprop{revpage}{}    \zref@addprop{main}{revpage}

\newcommand{\revi}[2]{%
	\zref@setcurrent{revsec}{\thesection}%
	\zref@setcurrent{revpage}{\thepage}%
	\zref@setcurrent{revcontent}{#2}%
	\refstepcounter{revc}%
	\label{#1}%
	\zlabel{#1}%
	\textcolor{blue}{#2}%
}
\newcommand{\revinu}[2]{%
	\zref@setcurrent{revsec}{\thesection}%
	\zref@setcurrent{revcontent}{#2}%
	\refstepcounter{revc}%
	\zlabel{#1}%
	\label{#1}%
	#2
}
\newcommand{\revr}[2]{%
	\zref@setcurrent{revsec}{\thesection}%
	\zref@setcurrent{revcontent}{#2}%
	\refstepcounter{revc}%
	\zlabel{#1}%
	\label{#1} \sout{#2}
}
\makeatother

\expandafter\def\expandafter\quote\expandafter{\quote\onehalfspacing\fontsize{12}{14}\selectfont}

\def\BibTeX{{\rm B\kern-.05em{\sc i\kern-.025em b}\kern-.08em
		T\kern-.1667em\lower.7ex\hbox{E}\kern-.125emX}}

\DeclareUnicodeCharacter{2212}{\textendash}
\DeclareUnicodeCharacter{03A8}{\textendash}

\begin{document}

\title{Minimum Mean Squared Error Holographic Beamforming for Sum-Rate Maximization}  
\author{Chandan~Kumar~Sheemar, Wali Ullah Khan, George C. Alexandropoulos,~\IEEEmembership{Senior Member,~IEEE}, \\Manzoor Ahmed, and Symeon Chatzinotas,~\IEEEmembership{Fellow,~IEEE}   
 \thanks{C. K. Sheemar, W. U. Khan, and S. Chatzinotas are with the SnT department at the University of Luxembourg (email\{chandankumar.sheemar, symeon.chatzinotas\}@uni.lu). G. C. Alexandropoulos is with the Department of Informatics and Telecommunications, National and Kapodistrian University of Athens, 16122 Athens, Greece (email: alexandg@di.uoa.gr). M. Ahmed, is with the School of Computer and Information Science and also with Institute for AI Industrial Technology Research, Hubei Engineering University, Xiaogan City 432000, China  (e-mails: manzoor.achakzai@gmail.com).} 
 
} 
 \maketitle
\begin{abstract}
 This paper studies the problem of hybrid holographic beamforming for sum-rate maximization in a communication system assisted by a reconfigurable holographic surface. Existing methodologies predominantly rely on gradient-based or approximation techniques necessitating iterative optimization for each update of the holographic response, which imposes substantial computational overhead. To address these limitations, we establish a mathematical relationship between the mean squared error (MSE) criterion and the holographic response of the RHS to enable alternating optimization based on the minimum MSE (MMSE). Our analysis demonstrates that this relationship exhibits a quadratic dependency on each element of the holographic beamformer. Exploiting this property, we derive closed-form optimal expressions for updating the holographic beamforming weights. Our complexity analysis indicates that the proposed approach exhibits only linear complexity in terms of the RHS size, thus, ensuring scalability for large-scale deployments. The presented simulation results validate the effectiveness of our MMSE-based holographic approach, providing useful insights.
\end{abstract}
\begin{IEEEkeywords}
Reconfigurable holographic surfaces, sum-rate maximization, MMSE, hybrid beamforming.
\end{IEEEkeywords}

\IEEEpeerreviewmaketitle

\section{Introduction} \label{Intro}

\IEEEPARstart{R}{econfigurable} holographic surfaces (RHS) have emerged as a key enabler for 6G, offering fine-grained control over electromagnetic waves through densely integrated sub-wavelength elements \cite{huang2020holographic,iacovelli2024holographic}. This feature enables holographic beamforming, which enhances spectral and energy efficiency, mitigates propagation challenges, and improves coverage in complex environments \cite{sheemar2025holographic}. Unlike traditional beamforming, holographic beamforming can be implemented with real-valued weights instead of phase modulation by leveraging the holographic principle \cite{sheemar2025secrecy,deng2021reconfigurable}.

Recent works on holographic beamforming for sum-rate maximization are available in\cite{deng2021reconfigurable,hu2023holographic,he2024linear,holtom2023holographic,suban2024beamforming}. For instance, \cite{deng2021reconfigurable} addresses the sum-rate maximization problem in RHS-assisted systems through a hybrid beamforming approach. In \cite{hu2023holographic}, the authors propose a fully analog holographic beamforming strategy that broadcasts a single data stream per user, aiming to enhance sum-rate performance, while minimizing the expense of additional radio frequency (RF) chains. The work in \cite{he2024linear} introduces an iterative optimization algorithm for sum-rate maximization for LEO satellite broadcasting. In \cite{holtom2023holographic}, distributed coherent mesh relay holographic beamforming based on the minimum-mean squared error (MMSE) criteria is presented. In \cite{suban2024beamforming}, an MMSE-based design for a passive holographic surface-assisted communication system is proposed.   

It is important to emphasize that, while MMSE-based designs have been proposed in \cite{holtom2023holographic, suban2024beamforming}, these approaches are not suitable for scenarios where the antenna front-end is implemented using an RHS. Moreover, existing approaches for sum-rate maximization
can be distinguished into two main categories. For instance,  \cite{he2024linear} employs a gradient ascent approach, whereas \cite{deng2021reconfigurable, hu2023holographic} reformulates the problem using approximations and introduces auxiliary variables, which further increases complexity. In both cases, each update of the holographic beamformer requires an iterative procedure which must be executed until convergence, thus linking the computational complexity not only to the size of the RHS but also to the number of inner iterations required for each update.  

In this letter, we study the problem of hybrid holographic beamforming for an RHS-assisted system. Our objective is to maximize the sum rate by recasting the problem in terms of the MMSE criterion \cite{christensen2008weighted}. This enables the development of the mean squared error (MSE)-based alternating optimization framework for holographic communications. A central result of our approach is the derivation of the direct relationship between the MSE measure and the holographic weights of the RHS. Unlike existing approaches that rely on iterative optimization for each update of the holographic beamformer, the proposed method capitalizes on the inherent quadratic dependence of the MSE on the holographic weights. This analytical insight enables the derivation of a closed-form optimal MMSE solution, effectively eliminating the need for iterative procedures for each update of the holographic beamformer. Our complexity analysis showcases that the proposed framework exhibits only linear complexity in terms of holographic elements, making it desirable for large-scale deployments. Simulation results are presented to validate the performance of the proposed MMSE-based design\footnote{\emph{Notations:} Scalars are denoted by regular letters, while vectors and matrices by bold lowercase and uppercase letters, respectively. $\mathbf{X}^\mathrm{T}$, $\mathbf{X}^\mathrm{H}$, and $\mathbf{X}^{-1}$ denote transposition, Hermitian transposition, and inverse, respectively.}.

\section{System Model and Design Objective}
\begin{figure}
    \centering
    \includegraphics[width=\linewidth,height=3.5cm]{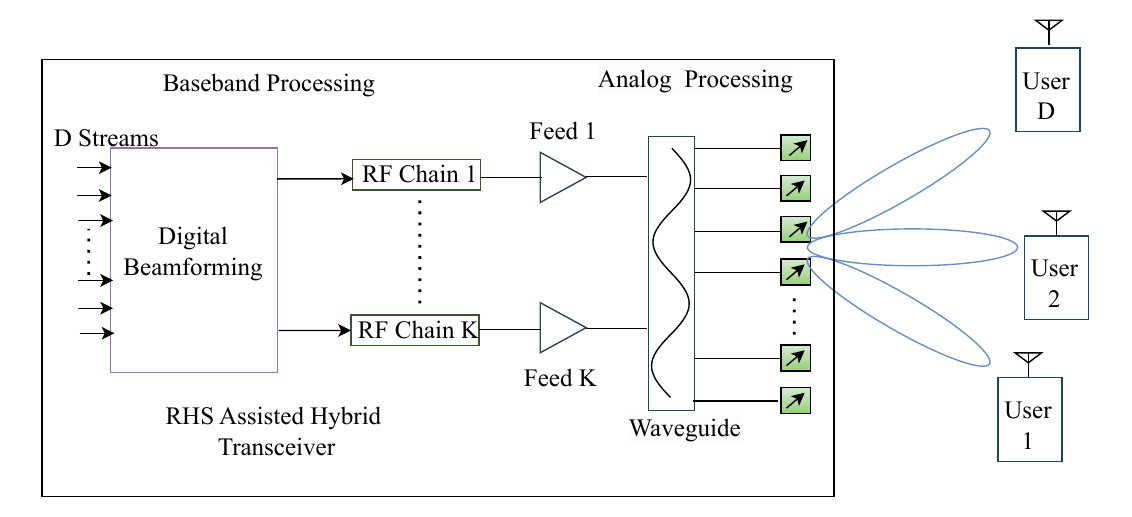}
    \caption{The considered hybrid RHS-assisted holographic multi-user communication system.}
    \label{fig_1}
\end{figure}

We consider a multi-user downlink holographic communication system, where a base station (BS) transmits data streams to $ D $ single-antenna users, represented by the set $ \mathcal{D} = \{1, \dots, D\} $. The BS is equipped with an RHS-based hybrid beamforming architecture, which processes signals in two stages: in baseband, via digital beamformers, and, in the analog domain, with holographic beamforming implemented via an RHS, as shown in Fig.~\ref{fig_1}. The processed signals in baseband are up-converted to the carrier frequency using $R$ RF chains, where $R$ is assumed to be equal to the number of feeds $K$ on the RHS. For the $ D $ active data streams to be transmitted, one for each user, let $ \mathbf{V} \in \mathbb{C}^{R \times D} $ denote the digital beamforming matrix in baseband, where each column $ \mathbf{v}_d \in \mathbb{C}^{R \times 1} $ represents the beamformer for downlink user $ d \in \mathcal{D}$. The number of RF chains in the hybrid RHS architecture is assumed to satisfy $ R=K \geq D $, ensuring that the system can support active data streams. Each RF chain feeds its corresponding RHS feed, with all feeds transforming the high-frequency electrical signal into a propagating electromagnetic reference wave. The radiation amplitude of the reference wave at each metamaterial element of the RHS is dynamically controlled using a holographic beamforming matrix $ \bmW \in \mathbb{C}^{M \times K} $, enabling the generation of desired directional beams, where $M$ denotes the number of elements of the holographic surface. Each element of the holographic matrix $\bmW$ is of the form 
$w_me^{-j \bmk_s \cdot \bmr_m^k}$ $\forall m=1,\ldots,M$ \cite{sheemar2025secrecy}, where $w_m$ denotes the holographic weight for beamforming for each $m$-th element and $e^{-j \bmk_s \cdot \bmr_m^k}$ represents the phase of the reference wave, with $\bmk_s$ denoting the the reference wave vector and $\bmr_m^k$ denoting the distance from the $k$-th RF feed to the $m$-th element of the RHS.

Let $y_d$ denote the received signal in baseband at the DL user $d \in \mathcal{D}$, which can be written as follows:
\begin{equation}\label{eq:received_signal}
\begin{aligned}
    y_d = &  \bmh_d^H \bmW \bmv_d s_d + \sum_{k \neq d}   \bmh_d^H \bmW \bmv_k s_k + n_d,
\end{aligned}
\end{equation}
where $\bmh_d$ denote the channel for user $d \in \mathcal{D}$, $n_d \sim \mathcal{CN}(0,\sigma_d^2)$ denote the additive white Gaussian noise with variance $\sigma_d^2$, and $s_d$ denotes data stream for user $d \in \mathcal{D}$, with $\mathbb{E}=[s_d s_d^*]=1$.
 
\subsection{Problem Formulation}
We aim to maximize the sum rate of the considered RHS-assisted multi-user system, which is expressed as follows:\begin{subequations}\label{problem_statement}
\begin{equation}
\max_{\bmV, \bmW}  \sum_{d \in \mathcal{D}}   \log_2\left( 1 + \frac{\left| \bmh_d^H \bmW \bmv_d \right|^2}{\sum_{k \neq d} \left| \bmh_d^H \bmW \bmv_k \right|^2 + \sigma_d^2} \right)
\end{equation} \label{WSR}
\begin{equation} \label{sum_pow}
\text{s.t.} \quad \mbox{Tr}( \bmW \bmV \bmV^H \bmW ) \leq  \alpha     
\end{equation}
\begin{equation}  \label{Holo_constriant}
 \quad \quad  \quad 0 \leq w_{m} \leq 1, \quad \forall m.
\end{equation}
\end{subequations}
where \eqref{sum_pow} denotes the total sum-power constraint and \eqref{Holo_constriant} denotes the holographic beamforming constraint \cite{sheemar2025secrecy}. It is noteworthy that the problem is inherently non-convex and very challenging to solve. Moreover, its formulation involves complex variables with real-domain constraints on $\bmW$, necessitating novel optimization strategies.

\section{MMSE-Based Reformulation and Optimization}
We leverage the relationship between MMSE and sum rate to reformulate the problem the sum-rate maxmization problem. To this end, we introduce auxiliary variables to enable a structured low-complexity alternating optimization framework for holographic beamforming.

\subsection{Reformulation of \eqref{problem_statement} via MMSE Reception}
We assume that each DL user $d$ deploys the weight $f_d$ in baseband to decode its data stream $s_d$. Therefore, we can write the estimated data stream for this user as follows: 
\begin{equation} \label{detect}
    \hat{s}_d = f_d^* y_d.
\end{equation}
Given \eqref{detect}, let the MSE for each user $d$ be denoted as $\text{MSE}_d =  \mathbb{E}[(\hat{s}_d-s_d) (\hat{s}_d-s_d)^*]$. By substituting \eqref{eq:received_signal} for the received signal model into this expression, it is deduced:
 \begin{equation} \label{ref_mse}
 \begin{aligned}
     \text{MSE}_d = & |f_d|^2 \bmh_d^H \bmW \bmv_d \bmv_d^H \bmW^H \bmh_d 
 + |f_d|^2 \sigma_d^2 
+ 1 
\\& + |f_d|^2 \sum_{k \neq d} \bmh_d^H \bmW \bmv_k \bmv_k^H \bmW^H \bmh_d - \hspace{-1mm}2 \text{Re}\left(f_d^* \bmh_d^H \bmW \bmv_d\right).
 \end{aligned}
\end{equation}
Hence, the optimal MMSE combiner per user $d$ is given as:
\begin{equation} \label{combiner_compt}
	\begin{aligned}
    f_d^* = \hspace{-1mm} &(
    \bmh_d^H \bmW \bmv_d \bmv_d^H \bmW^H \bmh_d  \hspace{-0.1mm}
    + \hspace{-0.5mm}  \sum\limits_{k \neq d}  
    \bmh_d^H \bmW \bmv_k \bmv_k^H \bmW^H \bmh_d
    + \sigma_d^2 )^{-1} \\& \times \bmh_d^H \bmW \bmv_d.
     \end{aligned}
\end{equation}
Given $f_d$ $\forall d$ and using the fundamental relationship between the sum rate and MSE~\cite{christensen2008weighted}, we obtain the following equivalent reformulation of \eqref{problem_statement}: 

\begin{equation}\label{MMSE_problem}
\min_{\bmV, \bmW, \textbf{m}}  \sum_{d \in \mathcal{D}}   m_d \text{MSE}_d \quad \text{s.t.} \quad \eqref{sum_pow} \;\text{and}\;\eqref{Holo_constriant}.
\end{equation}
where $\textbf{m} = [m_1,...,m_D]$ denotes the set of MSE weights, with $m_d >0$ $\forall d$.
By analyzing the gradients of \eqref{problem_statement} and \eqref{MMSE_problem}, it can be easily shown that the KKT condition for the sum rate and MMSE are equivalent if the weights $m_d$ are chosen as follows~\cite[eq. (24)]{cirik2015weighted}:
\begin{equation} \label{weight_calc}
    m_d = \frac{1}{\text{ln}(2)} \text{MSE}_d^{-1}.
\end{equation}

\subsection{Digital Beamforming Design}
Given the other variables fixed, \eqref{MMSE_problem} can be expressed withe respect to only the digital beamformer at the BS as:

\begin{equation}\label{problem_MSE_statement}
\min_{\bmV}   \sum_{d \in \mathcal{D}}  m_d \text{MSE}_d
\quad
\text{s.t.}  \quad \eqref{sum_pow}.
\end{equation}
To solve this problem, we calculate the partial derivative of the Lagrangian function of the objective function with respect to the conjugate of each $\bmv_d$. This yields the following optimal unconstrained digital beamformer:
\begin{equation}\label{digital_calc}
\begin{aligned}
\bmv_d^* =& ( m_d |f_d|^2 \bmW^H \bmh_d \bmh_d^H \bmW   + \sum_{k \neq d} \hspace{-1mm}m_k |f_k|^2 \bmW^H \bmh_k \bmh_k^H \bmW )^{-1} \\& \times m_d f_d^* \bmW^H \bmh_d.
\end{aligned}
\end{equation}
To meet the sum power constraint with equality, we normalize it to unit-norm and scale  it as $
\bmv_d = \sqrt{\frac{\alpha}{\text{Tr}(\bmW \bmV \bmV^H \bmW^H)}} \bmv_d^*$ $\forall d$.

\subsection{Holographic Beamforming Design}
In this section, we optimize the holographic beamformer assuming the other variables are fixed. The optimization problem with respect to $\bmW$ is written as follows: 
\begin{equation}\label{problem__holo_restatement}
\min_{\bmW}   \sum_{d \in \mathcal{D}}  m_d \text{MSE}_d
 \quad  \text{s.t.}   \quad \eqref{Holo_constriant}.
\end{equation}
The optimization of the holographic beamformer presents significant complexity, arising from the real-domain constraints and the intricate coupling among its various terms. To streamline the objective function, our initial approach focuses on elucidating the direct relationship between the MSE (and thus the sum rate) and the response of each reconfigurable element of the holographic beamformer. 
 
Note that we can decompose the holographic beamformer into two parts as $ \bmW = \text{diag}(\bmw) \mathbf{\Phi} $ \cite{sheemar2025secrecy,deng2021reconfigurable}, where $\bmw = [w_1, w_2, \dots, w_M]$ and $\mathbf{\Phi} $ contains fixed phase components, which describe the reference wave propagation within the waveguide. Given this observation, the subsequent result formally establishes our MSE relationship.

\begin{theorem}
In our hybrid RHS-assisted holographic multiuser
communication system, the sum-MSE corresponding to each $m$-th element of the holographic beamformer is given as
\begin{equation} \label{MSE_dependence}
    \begin{aligned}
    \text{MSE}(w_m) =&\sum_{d \in \mathcal{D}} m_d \bigg[  1 - \bigg(2 \text{Re}\bigg(f_d  {h}_{m,n}^{d*} w_m  \sum_{k=1}^K \mathbf{\Phi}_{m,k} v_d^k\bigg)\bigg) \\
&\quad + |f_d|^2 |{h}_{m,n}^{d*}|^2 |w_m|^2 \bigg|\sum_{k=1}^K \mathbf{\Phi}_{m,k} v_d^k\bigg|^2 \\
&\quad + |f_d|^2 |{h}_{m,n}^{d*}|^2 |w_m|^2 \sum_{l \neq d} \bigg|\sum_{k=1}^K \mathbf{\Phi}_{m,k} v_l^k\bigg|^2 \\
&\quad + |f_d|^2 \sigma_d^2 + c
\bigg].
    \end{aligned}
\end{equation} 
\end{theorem} 
\begin{proof}
    The proof is provided in Appendix~\ref{Appendix_1}.
\end{proof}
Based on the results above, the optimization problem can be reformulated in terms of each element $w_m$ of the holographic beamformer as follows:
\begin{equation}\label{problem__holo_restatement}
\min_{w_m} \sum_{d \in \mathcal{D}}  m_d \text{MSE}_d
\quad 
 \text{s.t.}  \quad  \eqref{Holo_constriant}.
\end{equation}
In the following, we initially disregard the constraint  \eqref{Holo_constriant}  and derive a fundamental result for the optimal unconstrained updates without requiring iterations.

\begin{lemma}
   The optimal unconstrained closed-form update for each $w_m$ is given as 
   \begin{equation} \label{opt_value}
w_m = 
\frac{\sum\limits_{d \in \mathcal{D}} m_d \text{Re} \bigg(f_d {h}_{d,m}^{*} 
\sum\limits_{k=1}^K \mathbf{\Phi}_{m,k} v_d^k \bigg)}
{\sum\limits_{d \in \mathcal{D}} m_d |f_d|^2 \bigg[
|{h}_{d,m}^{*}|^2 \bigg|\sum\limits_{k=1}^K \mathbf{\Phi}_{m,k} v_d^k\bigg|^2 
+ d \bigg]}
\end{equation}
where $d =|{h}_{d,m}^{*}|^2 \sum\limits_{l \neq d} \bigg|\sum\limits_{k=1}^K 
\mathbf{\Phi}_{m,k} v_l^k \bigg|^2$.
\end{lemma}
\begin{proof}  
To prove this result, let $f(\cdot)$ denote our objective function. Note that, as the MSE function exhibits a quadratic dependence on \( w_m \), it can be expressed in the standard form  
$  f(w_m) = a_m w_m^2 - 2 b_m w_m + c_m,  
$
where \( a_m > 0 \) ensures convexity, which is our case. To find the optimal value of \( w_m \), we take the first derivative and set it to zero $
   \partial f(w_m)/\partial w_m = 2 a_m w_m - 2 b_m = 0. $
Solving for \( w_m \), we obtain the optimal closed-form solution $ w_m^* = b_m/a_m.$ 
Since \( f(w_m) \) is convex, this solution is guaranteed to be a global minimum. 
\end{proof}
To satisfy the constraint  \( 0 \leq w_m \leq 1 \), we project its value on the feasible set as follows:
\begin{enumerate}
    \item If the unconstrained solution satisfies \( 0 \leq w_m \leq 1 \), then no adjustment is needed.
    \item If the unconstrained solution is negative (\( w_m < 0 \)), it is projected to the lower bound \( w_m = 0 \).
    \item If the unconstrained solution exceeds the upper bound (\( w_m > 1 \)), it is projected to the upper bound \( w_m = 1 \).
\end{enumerate}
  
Based on the aforementioned projection, we can derive the optimal closed-form constrained updates for $w_m$ as follows:
\begin{equation}
w_m = \min\Big( \max\Big( w_m, 0  \Big),\ 1 \Big).
\end{equation}

\emph{Remark:}  
The closed-form updates for $\bmW$, derived from its quadratic dependence on the MSE, eliminate the necessity for iterative optimization until convergence for each update of \( \bmW \), thereby significantly reducing computational complexity.  In contrast, iterative approaches may require fine-tuning of multiple parameters \cite{deng2021reconfigurable,he2024linear, hu2023holographic}, and improper tuning may lead to convergence at a sub-optimal solution, whereas the proposed method directly guarantees an optimal update without such dependencies.

\subsection{On the Convergence}
The proposed MMSE-based alternating optimization algorithm iteratively updates the digital beamforming matrix \( \bmV \), the holographic beamforming matrix \( \bmW \), and the MSE weights \( \textbf{m} \). Each step ensures a monotonic decrease in the objective function. The digital beamforming update in \eqref{digital_calc} minimizes the sum-MSE for a fixed \( \bmW \), while the holographic beamforming update in \eqref{opt_value}, derived from the quadratic dependence of the MSE on \( w_m \), guarantees an optimal closed-form solution. Additionally, the MSE weight update in \eqref{weight_calc} maintains the equivalence between MMSE and um rate, preserving alignment with the original objective function. Since the sum-MSE function is non-negative and lower-bounded, the optimization sequence cannot decrease indefinitely. Furthermore, our overall approach, summarized in Algorithm~\ref{algmmse_beamforming}, follows a block coordinate descent framework, where each subproblem is solved optimally in closed form, ensuring that updates move toward a stationary point. Therefore, the proposed algorithm is guaranteed to converge to, at least, a locally optimal solution. The convergence behaviour of the proposed method is shown in Fig.~\ref{conv}.

 \begin{algorithm}[t]
\caption{Proposed Hybrid Holographic Beamforming}
\label{algmmse_beamforming}
\begin{algorithmic}[1]
\State \textbf{Input} Initial $\mathbf{V}$, $\bmW$, and set tolerances $\epsilon$ and $T_{max}$.
\Repeat
    \State Update the MMSE reception weights via \eqref{combiner_compt}.
    \State Update the MSE weights with \eqref{weight_calc}.
    \State  Update the MMSE digital beamformers with \eqref{digital_calc}.
    \State  Update the MMSE holographic beamformer via \eqref{opt_value}.
    \State  Increment iteration counter $t \gets t + 1$.
\Until{Convergence or $t \geq T_{max}$}.
\State  \textbf{Output} Optimized $\mathbf{V}$ and $\mathbf{W}$.
\end{algorithmic}
\end{algorithm}

\subsection{Complexity Analysis}
The per-iteration computational complexity of the proposed algorithm is primarily determined by the two stages of digital and holographic beamforming. The digital beamformer optimization requires inverting a \(K \times K\) matrix per user, resulting in an overall complexity of \(\mathcal{O}(D K^3)\), while the scaling factor to meet the power constriant adds \(\mathcal{O}(D K)\), which is negligible. For holographic beamforming, computing the closed-form solution for \(w_m\) via \eqref{opt_value} across \(M\) elements involves summations over \(D\) users and \(K\) feeds, leading to \(\mathcal{O}(M D K)\). Therefore, the overall complexity is \(\mathcal{O}(D K^3 + M D K)\), which in only linear in $M$.

 The only work addressing the considered scenario is \cite{deng2021reconfigurable}, which lacks complexity analysis. However, therein, due to the iterative updates of \(\bmW\) and the introduction of auxiliary variables, the complexity is expected to be at least \(\mathcal{O}(I (M D K + C))\), where \(I \gg 1\) represents the required iterations per update of \(\bmW\), and \(C\) accounts for the number of the auxiliary variables. The only algorithm with the reported linear complexity is \cite{he2024linear}, but it still depends on the number of iterations and assumes a fully analog system with no digital beamforming. In hybrid beamforming, gradient updates of the holographic beamformer are required until convergence for each update of the digital beamformers. Moreover, the gradient computation becomes highly complex (see, e.g., \cite[eq. (21)]{sheemar2025joint_QoS}), increasing the overall complexity to \(\mathcal{O}(I M^2)\).

\section{Numerical Results and Discussion} 
In this section, we evaluate the performance of the proposed method. We consider a system with \(D = 3\) users and \(K = 6\) RF chains. The signal-to-noise ratio (SNR) is defined as the transmit SNR, with the transmit power normalized to $1$, while the noise variance, \(\sigma_d^2\), is adjusted to achieve the desired SNR. The carrier frequency is set to 30 GHz and the element spacing is assumed to be \(\lambda/4\).  In free space, the propagation vector is represented as \(\mathbf{k}_f\), while on the RHS, it is denoted by \(\mathbf{k}_s\). According to electromagnetic theory, the magnitude of \(\mathbf{k}_s\) is related to \(\mathbf{k}_f\) by  
$|\mathbf{k}_s| = \sqrt{\epsilon_r} |\mathbf{k}_f|$,
where \(\epsilon_r\) represents the relative permittivity of the RHS substrate, which is typically around 3. We adopt standard values from the literature, setting \(|\mathbf{k}_s| = 200 \sqrt{3} \pi\) and \(|\mathbf{k}_f| = 200 \pi\). Since the system operates at 30 GHz, we model the channel $\bmh_d$ as a pathwise channel with \(I = 5\) paths, and the RHS response as a uniform planar array (UPA). To evaluate the performance of the proposed method, we considered two benchmark schemes: 1) The method proposed in \cite{deng2021reconfigurable} (Benchmark 1), and 2) Random selection of the $\bmW$ values (Benchmark 2).
 
\begin{figure}
    \centering
    \includegraphics[width=0.6\linewidth]{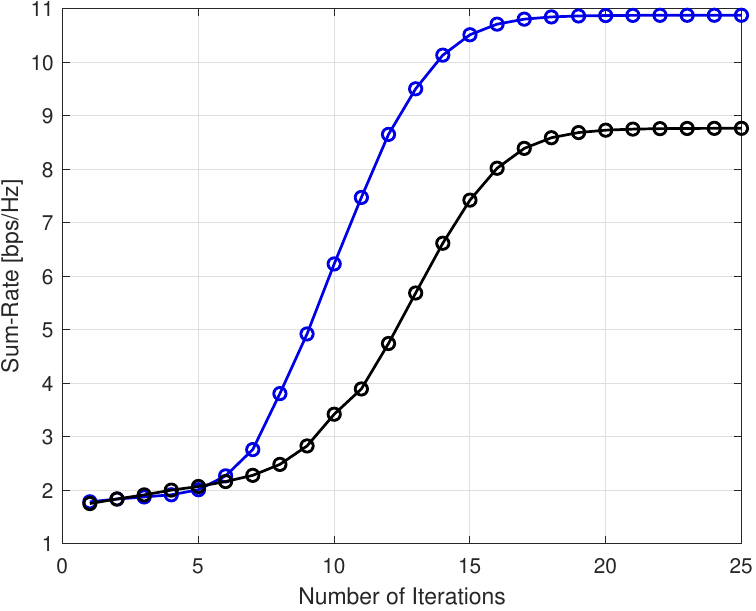}
    \caption{Convergence of the proposed MMSE-based hybrid holographic beamforming scheme for SNR= 0 dB.}
    \label{conv}
\end{figure}

 \begin{figure}
      \centering
        \includegraphics[width=0.6\linewidth]{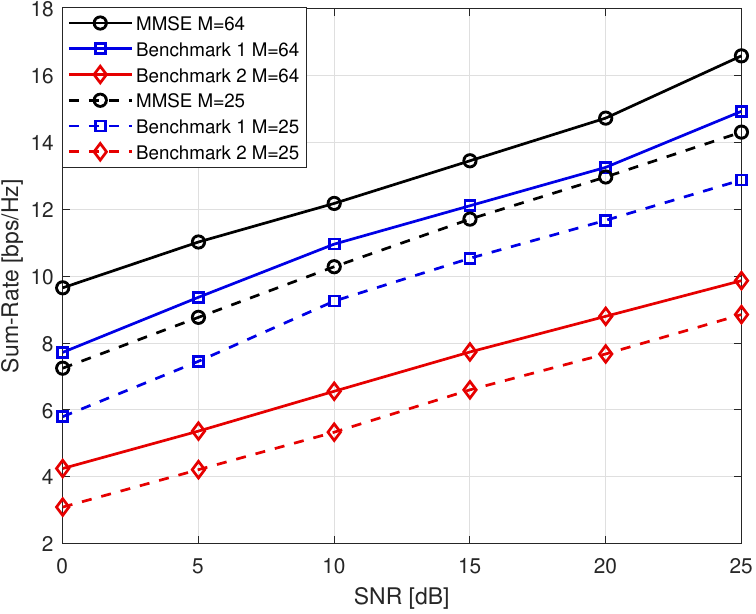}
        \caption{Achievable sum rate versus SNR for $M=\{25,64\}$.}
        \label{figwmmse_snr}
 \end{figure}

    %
     %

In Fig.~\ref{figwmmse_snr}, the performance of the proposed method is analyzed as a function of the SNR for RHS sizes \(M = \{25,64\}\). The results indicate that, in both cases, the proposed method consistently outperforms the benchmark schemes, achieving a higher sum-rate performance across the entire SNR range. This superior performance is attributed to the robustness of our optimization approach, which enables optimal closed-form updates, eliminating the need for extensive fine-tuning of multiple parameters to converge to a local optimum. The impact of increasing \(M\) is also evident, as a larger RHS size leads to significant performance gains across all methods, primarily due to enhanced holographic beamforming gain and improved interference suppression. Furthermore, the proposed framework not only outperforms both benchmark schemes, but also maintains high computational efficiency, as its complexity scales linearly with the RHS size, making it well-suited for large-scale implementations.
\begin{figure}
    \centering
        \includegraphics[width=0.6\linewidth]{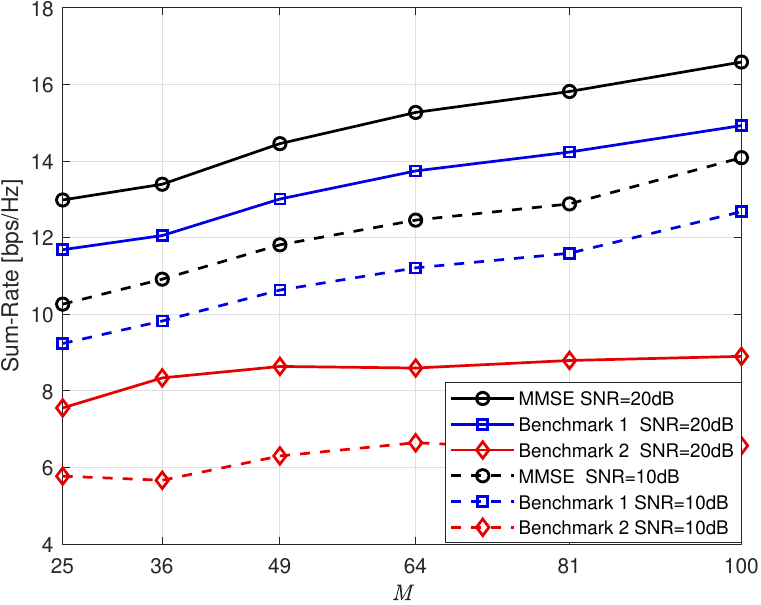}
        \caption{Achievable sum rate versus $M$ for SNR=$\{10,20\}$dB.}
        \label{figwmmse_m}
\end{figure}

In Fig.~\ref{figwmmse_m}, the performance as a function of the RHS size $M$ is shown at the SNR levels $10$ dB and $20$ dB. The results demonstrate a clear performance improvement as $M$ increases for all the schemes, highlighting the benefits of a larger RHS in jointly enhancing beamforming gain and interference suppression. At both SNR levels, it is visible that the proposed method consistently outperforms the benchmark schemes, demonstrating its effectiveness in maximizing the system's sum-rate performance.

\section{Conclusion}
In this letter, the problem of hybrid holographic beamforming for sum-rate maximization in RHS-assisted multi-user communication systems was investigated. By establishing a mathematical link between the MMSE criterion and the sum rate, a closed-form optimal solution was derived, eliminating the need for iterative optimization. It was demonstrated that the proposed approach achieves linear complexity with respect to the RHS size, ensuring scalability for large-scale deployments. Our simulation results validated our approach's superior performance and efficiency over benchmark schemes.
 
\begin{appendices}
 \section{MSE Derivation} \label{Appendix_1}
To establish the relationship between the MSE and each $w_m$ comprising $\bmw$, we consider highlighting their dependence. Recall that the $ \text{MSE}_d$ is given as follows:
\begin{align*}
\text{MSE}_d &= 1 \underbrace{- 2\text{Re}\big(f_d \bmh_d^H  \bmW \bmv_d\big)}_{=x_2}
+ \underbrace{|f_d|^2 |\bmh_d^H \bmW \bmv_d|^2}_{=x_3}
\\&\quad + \underbrace{\sum_{l \neq d} |f_d|^2 |\bmh_d^H \bmW \bmv_l|^2}_{=x_4} + |f_d|^2 \sigma_d^2.
\end{align*}
The \(m\)-th element of \(\bmw\) affects \(\text{MSE}_d\) through its contribution to terms involving \(\bmW \bmv_d\) and \(\bmW \bmv_l\) for \(l \neq d\). The real part of the desired signal term is given in $x_2$.
By expanding the contribution of \(w_m\), each \(m\)-th element of \(\bmW \bmv_d\) can be written as $
(\bmW \bmv_d)_m = w_m \sum_{k=1}^K \mathbf{\Phi}_{m,k} v_d^k$,
where \(v_d^k\) is the \(k\)-th element of the beamformer vector \(\bmv_d\) and $\mathbf{\Phi}_{m,k}$ denotes the $k$-th element in the $m$-th row. Substituting this into \(x_2\), yields
\begin{equation}
x_2 = -2 \text{Re}\bigg(f_d {h}_{d,m}^{*} \cdot w_m \sum_{k=1}^K \mathbf{\Phi}_{m,k} v_d^k \bigg) + c_2,
\end{equation}
where \({h}_{d,m}^{*}\) is the conjugate of the channel between the \(m\)-th element of the holographic surface and the downlink user \(d\), and $c_2$ contains constant terms not depending on $w_m$.

The third term $x_3$ denotes the norm-squared of the desired signal term.
Expanding the inner product leads to 
$ 
|\bmh_d^H \bmW \bmv_d|^2 = \big(\bmh_d^H \bmW \bmv_d\big)^H \big(\bmh_d^H \bmW \bmv_d\big)$, 
and by noting that $
\bmh_d^H \bmW \bmv_d = \sum_{m} {h}_{d,m}^{*} (\bmW \bmv_d)_m$
with \((\bmW \bmv_d)_m = w_m \sum_{k=1}^K \mathbf{\Phi}_{m,k} v_d^k\) indicating the terms involving only $w_m$, we substitute
$ \bmh_d^H \bmW \bmv_d = {h}_{d,m}^{*} w_m \sum_{k=1}^K \mathbf{\Phi}_{m,k} v_d^k+ c_3,
$
where $c_3$ denotes constant terms which do not depend on $w_m$. Taking the norm-squared and multiplying with \(|f_d|^2\), the following simplication is obtained:
\begin{equation}
x_3 = |f_d|^2 |{h}_{d,m}^{*}|^2 |w_m|^2 \bigg|\sum_{k=1}^K \mathbf{\Phi}_{m,k} v_d^k\bigg|^2.
\end{equation}

The last term $x_4$ can be decomposed similarly, leading to: 
\begin{equation}
x_4 = |f_d|^2 |{h}_{d,m}^{*}|^2 |w_m|^2 \sum_{l \neq d} \Big(\bigg|\sum_{k=1}^K \mathbf{\Phi}_{m,k} v_l^k\bigg|^2 + c_l \Big),
\end{equation}
where $c_l$ includes constant terms.
Combining all the terms together with the decomposed structure, we get the final MSE expression for each user $d$ as follows:
\begin{equation}
    \text{MSE}_d =  1 - x_1 + x_2 +x_3 +x_4 + |f_d|^2 \sigma_d^2 + c,
\end{equation}
 where $c$ contains all the constants which do not depend on $w_m$, which concludes the proof.
\end{appendices}

\ifCLASSOPTIONcaptionsoff
  \newpage
\fi

{\footnotesize
\bibliographystyle{IEEEtran}
\def\baselinestretch{0.9}
\bibliography{main}}
  
\end{document}